\documentclass[12pt]{article}
\usepackage[utf8]{inputenc}
\usepackage[english]{babel} 
\usepackage{amsthm}
\usepackage{amsmath} % for math formatting
\usepackage{pgfplots} % for creating figures
\usepackage{tikz}
\usepackage{sgame} % for game theory tables.
\pgfplotsset{compat=1.7}%the version of pgfplots to use
\usetikzlibrary{intersections, calc}
\usetikzlibrary{patterns, fillbetween}
\usepackage{natbib}
\usepackage{amsfonts}
\usepackage{bbm}
\usepackage{comment} %commenting bubbles package.
\usepackage{url}
\usepackage[doublespacing]{setspace}
\usepackage{lipsum}
\usepackage{amssymb}
\usepackage[normalem]{ulem}
\usepackage{subfigure}
\usepackage{array}
\usepackage{hyperref}
\usepackage[title,titletoc]{appendix}

\usepackage{todonotes}

\usetikzlibrary{scopes,patterns,intersections,calc}

\newtheorem{theorem}{Theorem}

\newtheorem{corollary}{Corollary}

\newtheorem{definition}{Definition}
\newtheorem{example}{Example}

\newtheorem{lemma}{Lemma}

\newtheorem{proposition}{Proposition}

\newcommand{\argmax}{{\rm argmax}}

\newcommand{\dN}{{\bf N}}

\newcommand{\dR}{{\bf R}}

\newcommand{\calZ}{{\cal Z}}
\newcommand{\ep}{\varepsilon}

\def\AM{A}
\def\PM{P}
\def\UM{u}

\title{Mechanism Design with Spiteful Agents\thanks{Lagziel acknowledges the support of the Israel Science Foundation, Grant \#2074/23.
Solan acknowledges the support of the Israel Science Foundation, Grant \#211/22.}}

\author{
        Aditya Aradhye\thanks{{Economics Department, Ashoka University, India. E-mail: \textsf{adityaaradhye@gmail.com}.}}
        \and
        David Lagziel\thanks{Department of Economics, Ben-Gurion University of the Negev, Israel.  E-mail: \textsf{Davidlag@bgu.ac.il}. }
        \and
        Eilon Solan\thanks{The School of Mathematical Sciences, Tel Aviv University, Israel. E-mail: \textsf{eilons@tauex.tau.ac.il}.}}
\date{\today}

\begin{document}

\maketitle

\thispagestyle{empty}

\begin{abstract} \singlespacing{
We study a mechanism-design problem in which spiteful agents strive to not only maximize their rewards but also, contingent upon their own payoff levels, seek to lower the opponents' rewards.
We characterize all individually rational (IR) and incentive-compatible (IC) mechanisms that are immune to such spiteful behavior, showing that they take the form of threshold mechanisms with an ordering of the agents.  
Building on this characterization, we prove two impossibility results: under either anonymity or efficiency, any such IR and IC mechanism collapses to the null mechanism, which never allocates the item to any agent.
Leveraging these findings, we partially extend our analysis to a multi-item setup.
These results illuminate the challenges of auctioning items in the natural presence of other-regarding preferences.}
\end{abstract}

\bigskip


\bigskip


\newpage

\setcounter{page}{1} 
\section{Introduction}

On August 12, 2020, the Israeli Ministry of Communications announced the surprising results of the 5G spectrum auction held earlier that month, in which three groups participated -- Cellcom, Pelephone, and Partner. 
Although all telecommunications groups secured bandwidth bundles enabling 5G operations, the Cellcom group was required to pay $30\%$ more than the Pelephone group despite receiving an inferior bundle.\footnote{The Cellcom group also paid $80\%$ more than the Partner group, though the latter did not secure a superior bundle.}
That evening, the CEO of the Cellcom group tweeted an explanation for the seemingly poor outcome, stating: ``We chose not to hurt the others, and they chose to hurt us; it's as simple as that!"
Interestingly, the Cellcom group challenged the 5G spectrum auction in real time and appealed to the Administrative Court, claiming that the auction design allowed for manipulation and spiteful bidding. 
The appeal was ultimately rejected.  %https://x.com/orengabayg/status/1293628839055306754

The phenomenon of spiteful bidding is not unique to the Israeli 2020 spectrum auction.
In the Swiss 2012 spectrum auction, Sunrise paid \(34\%\) more than Swisscom for an inferior bundle.  
Similarly, in the Austrian 2013 spectrum auction, revenues were much higher than expected due to highly aggressive bidding in the sealed-bid stage. This led Georg Serentschy, the Managing Director of the Telecommunications and Postal Services Division of the Austrian Regulatory Authority for Broadcasting and Telecommunications, to remark:  
``In the opinion of the regulatory authority, the price of EUR 2 billion, which was surprisingly high for us, is to be attributed to the consistently offensive strategy followed by the bidders.''\footnote{Earlier examples of potential spiteful bidding in spectrum auctions can be found in the studies of \cite{Borgers2005, Brandt2007, Maasland2007}.}
These auctions and public comments mark the starting point of our study: a mechanism design problem with spiteful agents.

This study builds upon the notion of spiteful agents---those who not only aim to maximize their rewards but also, given their own payoff levels, seek to minimize the rewards of their opponents. 
Such other-regarding preferences are quite natural, even expected, in scenarios where the competitive interaction among agents extends beyond a single auction. 
A prime example is spectrum auctions, where participants also compete in the telecommunications market. 
However, this dynamic is not limited to telecommunications; it is equally applicable to construction projects, retail businesses, the food industry, and any other economic sector where agents compete for advantage in external markets.

The main analysis considers a single-item, private-value auction where the agents strive to maximize their own payoffs, and \emph{conditional on their own rewards}, exhibit a spitefulness property. 
We refer to this property as \emph{spitefulness}, meaning that each agent prefers to reduce the payoffs of some other agents provided that no agent is made better off.
We embed this property into the solution concept of a \emph{Spite-Free Nash Equilibrium} (SNE), in which no player can increase her own payoff through a unilateral deviation, nor can they reduce the payoffs of others while keeping her own payoff fixed.
Within this framework, we aim to characterize mechanisms that are Incentive Compatible (IC) and Individually Rational (IR) while admitting an SNE.

Our first and main result provides a characterization that links the aforementioned properties to threshold mechanisms.  
Formally, a threshold mechanism consists of an ordering of the agents and an individual threshold value assigned to each agent.
According to their predetermined ordering, the agents are sequentially offered take-it-or-leave-it offers to purchase the item at a price which is equal to their threshold.
The first agent to accept the offer receives the item; if no agent  accepts, the item remains unallocated.%
\footnote{To be precise, if no bid is above the respective threshold, then the item is allocated according to some tie-breaking rule among the agents whose bid coincides with their threshold, or is not allocated.
If all bids are below the respective thresholds, the item is not allocated.}
We show that a mechanism is IR and IC with an SNE if and only if it is a threshold mechanism.

It is straightforward to see why threshold mechanisms satisfy the stated properties. 
Since payments are predetermined and offers are made sequentially, agents have no ability to bid spitefully in a way that reduces others’ expected payoffs without also decreasing their own. 
Figure \ref{fig:Thresholdexample} provides some visual intuition for this result by depicting a threshold mechanism in a two-agent setting. 
By contrast, the second-price mechanism does not admit a truthful SNE
as non-winning agents can spitefully raise their bids to reduce the expected payoff of the winning agent (as was shown in \cite{Brandt2002} and \cite{Morgan2003}, among others). 
The more challenging part, however, is to show that \emph{no other mechanism} satisfies IR, IC, and has an SNE.
\color{black}

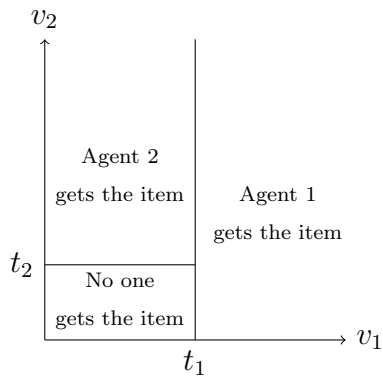
\begin{figure}[h!tbp]
    \centering
    \begin{tikzpicture}[cross/.style={path picture={\draw[black]
(path picture bounding box.south east) -- (path picture bounding box.north west) (path picture bounding box.south west) -- (path picture bounding box.north east);}}]
\tikzstyle{dot}=[circle,draw,inner sep=1.5,fill=black]
\draw[->] (0,0) -- (4, 0) node[right] {$v_1$};
\draw[->] (0,0) -- (0, 4) node[above] {$v_2$};
\draw[-] (2, 4) -- (2,0) node[below] {$t_1$};
\draw[-] (2, 1) -- (0,1) node[left] {$t_2$};
\node[label={[align=center, xshift=1.1cm, yshift=-1cm]\scriptsize Agent 1\\ \scriptsize gets the item}] at (2, 2) {};
\node[label={[align=center, yshift=0.5cm]\scriptsize Agent 2\\ \scriptsize gets the item}] at (1, 1) {};
\node[label={[align=center, yshift=-0.6cm]\scriptsize No one\\ \scriptsize gets the item}] at (1, 0.45) {};
\end{tikzpicture}
\caption{\footnotesize A two-agent threshold mechanism represented in the valuation plane. The thresholds of the two agents are $t_1$ and $t_2$, respectively.
If agent \(1\) bids above her threshold \(t_1\), she receives the item and pays \(t_1\). Otherwise, the item is allocated to agent \(2\), provided that her bid exceeds \(t_2\), in which case she pays \(t_2\).}
\label{fig:Thresholdexample}
\end{figure}

We use the characterization to establish two impossibility results, both of which concern the \emph{null mechanism}, under which the item is never allocated and no agent pays anything, regardless of the agents’ bids. 
The first result shows that any IC, IR, and anonymous (i.e., symmetric) mechanism that admits an SNE must be the null mechanism. 
The second result shows that any IC, IR, and efficient mechanism\footnote{Efficiency here means that no losing agent bids above the winning agent.} that admits an SNE must also be the null mechanism.  
The intuition behind both results is straightforward in light of our characterization. 
A symmetric threshold mechanism must be null in order to eliminate the role of the initial ordering of agents.
Similarly, an efficient threshold mechanism must be null to eliminate the possibility that agents with lower valuations are ranked ahead of agents with higher valuations.

We also extend the notion of SNE to the multi-item setup.
A bid profile is an SNE in this setup if 
there is no unilateral deviation of any agent that
(a) increases that agent's payoff, or
(b) keeps that agent's payoff fixed, does not increase the payoff of any other agent,
and strictly decreases the payoff of at least one other agent.
While the characterization of IR mechanisms where truthful reporting is an SNE is beyond our reach,
we are able to provide some guidance and insights on their structure.

\subsection{Related work}

A growing body of research examines the impact of spiteful or `antisocial' preferences on auction outcomes. 
Early contributions include \cite{Morgan2003}, who derive symmetric equilibria for several common auction mechanisms under the assumption that bidders attach disutility to the surplus of their rivals. 
They demonstrate that standard equivalence results between auction formats no longer hold in this setting, and that spitefulness induces more aggressive bidding relative to the classical framework without such preferences. 
Related work by, e.g., \cite{Brandt2002,Brandt2007,Vetsikas2007,Tang2012}, and  \cite{Chen2016a} develops models in which spiteful agents maximize a weighted difference between their own profit and that of their competitors, again leading to more aggressive bidding behavior. 
In addition, \cite{Brandt2002} explores how repeated interactions may enable bidders to infer one another's private valuations, while \cite{Brandt2007} compare the symmetric Bayes-Nash equilibria of first- and second-price auctions. 
The latter study shows that expected revenue in second-price auctions is higher under spiteful bidding, converging to equivalence across formats only in the extreme case where agents exclusively care about reducing rivals’ profits. 
Extending this line of research, \cite{Vetsikas2007} consider auctions with multiple identical items where each agent can win at most one unit, and show that the revenue-maximizing choice between $m$- and $(m+1)$-price auctions depends on the degree of spitefulness.

Beyond these models, \cite{Maasland2007} study the equilibrium behavior of spiteful agents when they care about the \emph{amount paid by the winner} rather than the winner’s profit. 
They refer to this as a case of ``financial externalities'' and analyze first- and second-price auctions, with and without reserve prices.
Their analysis shows that such externalities reduce expected prices in the first-price format but have an ambiguous effect in second-price settings.\footnote{This framework connects naturally to a broader literature on auctions with externalities, including \cite{Jehiel1996,Jehiel1996b,Jehiel1999a,Jehiel2000}, among many others.} 

Another strand of the literature examines spiteful bidding in combinatorial settings. \cite{Janssen2016,Janssen2019} show that truthful bidding does not constitute an equilibrium in combinatorial auctions, and that agents’ types are not fully revealed even in efficient outcomes. 
As do we, their model also represents spitefulness via lexicographic preferences, but in terms of raising rivals’ costs. 
\cite{Gretschko2016} extend this discussion to practical applications, showing how spiteful bidding can increase the complexity of strategies in combinatorial clock auctions, where truthful bidding may become suboptimal due to other-regarding preferences.

Several papers examine extensions and variations on these themes. 
\cite{Sharma2010} consider asymmetric environments in which agents differ in their degree of spitefulness, finding that equilibria can yield inefficient allocations and alter the revenue ranking between first- and second-price formats. 
\cite{Zhou2007} focus on keyword auctions run by search engines such as Google, where bidders with fixed private values may exhibit ``vindictive'' behavior. 
They show that in such environments a pure-strategy equilibrium may fail to exist.

Finally, experimental work provides evidence in line with these theoretical predictions. 
\cite{Cooper2008} and \cite{Kimbrough2012},
for example, document that in controlled laboratory settings agents who perceive their opponents as having substantially higher valuations tend to overbid, thereby reducing their opponents’ expected payoffs.

\subsection{The structure of the paper}
The paper is organized as follows. 
Section~\ref{Section - the model} introduces the model and basic definitions. 
Section~\ref{sec:impossibility} presents the main characterization and two impossibility results.
Section~\ref{Section: optimal spite free mechanism - example} discusses the efficiency of threshold mechanisms, and Section~\ref{sec:Multiple} studies the multi-item setup.

\section{The Model} \label{Section - the model}

We study the problem of assigning a single indivisible item among $n$ agents $I = \{1,\ldots,n\}$, where each agent $i \in I$ has a private valuation $v_i \in \dR_+$ for the item. 
Let $V = \dR_+$ denote the set of feasible valuations, and further assume that the reservation value of each agent, conditional on not receiving the item, is $0$.
A \emph{bid} $b_i \in \dR_+$ of agent $i$ is the agent's reported valuation, and a \emph{bid profile} $b = (b_1,\ldots,b_n) \in \dR_+^n$ is a tuple of bids, one  for each agent. 
When relating to agent $i$, 
we follow the standard notation of $b = (b_i,b_{-i})$.

An \emph{allocation} determines the assignment of the item. 
The set of allocations is $\calZ = I \cup \{0\}$, so that the allocation $i\in I$ corresponds to the item being assigned to agent $i$, and the allocation $0$ corresponds to the item being unassigned.
Agent $i$'s \emph{payment} is a non-negative real number. 
A \emph{mechanism} $M = (\AM,\PM)$ consists of an allocation function $\AM:V^n \to \calZ$ and a payment function $\PM:V^n \to \dR_+^n$. 
The allocation function determines the allocation given a bid profile $b$, and
the payment function determines the payments vector given $b$.
Abusing notations, 
for every agent $i$ we denote by $\AM_i:V^n \to \{0,1\}$
the function which is equal to $1$ if agent $i$ gets the item, and $0$ otherwise.
We denote by $\PM_i:V^n \to \dR_+$ agent $i$'s payment function.

For a given mechanism $M$, agent $i$'s utility function $\UM_i: V^n \times V \to \dR$ depends on the bid profile and on agent~$i$'s private value, and is defined by 
\[
\UM_i(b;v_i) := v_i \cdot \AM_i(b) - \PM_i(b).
\]
In case $b_i\neq v_i$, we say that agent $i$ \emph{misreports} (his valuation).

Our use of $\dR_+$ as the agents’ action space involves some loss of generality compared to a more general, potentially multi-dimensional, action space.
Nonetheless, this assumption is quite natural in practical settings where a single item is being allocated. 
In Section~\ref{sec:Multiple}, where we extend the model to the allocation of multiple items, the action space is generalized accordingly.

\subsection{Simple properties of mechanisms}

In this section we list several well-known properties of mechanisms, namely anonymity, efficiency, individual rationality, and incentive compatibility, that will be used throughout the analysis and characterization.

A mechanism \(M = (A,P)\) is \emph{anonymous} if its outcome is independent of the agents’ indices. Formally, for every permutation \(\pi\) on $I$ and every bid profile \(b\),  
\[
\big(\AM_1(\pi (b)),\dots,\AM_n(\pi (b))\big) = \pi \big(\AM_1(b),\dots,\AM_n(b)\big), 
\quad \text{and} \quad 
\PM(\pi (b)) = \pi \big(\PM(b)\big).
\]  
A mechanism is \emph{efficient} if the winning agent’s bid is at least as high as every other bid. 
Formally, for every bid profile \(b\) such that \(\AM(b) = i\) for some agent \(i\), we have  \(b_i \ge b_j\)  for every \(j \in I\). 
The mechanism is \emph{Individually Rational} (IR) if every agent can secure the reservation value (normalized to $0$) by bidding truthfully: $\UM_i(b;v_i) \geq 0$ for every agent $i$, every private valuation $v_i$, and every bid profile $b$ such that $b_i=v_i$.
The mechanism is \emph{Incentive Compatible} (IC) if for every valuation profile $v$, the profile $b=v$ of truthful bids is a Nash equilibrium.

The following result lists standard properties of IR and IC mechanisms:
(a) if the mechanism is IR, then an agent who did not receive the item pays 0;
(b) if the mechanism is IC, then given the bids of all non-winning agents, the payment of the winning agent~$i$ is independent of her bid, conditional on winning;
and
(c) if the mechanism is IR and IC, then any agent who receives the item under a given bid profile will also receives the item if she increases her bid; For proofs, see, e.g., \cite{myerson1981optimal} or \cite{Krishna2009}.

\begin{lemma}
\label{lemma:simple properties}
\begin{enumerate}
\item\label{lem:ZeroPayment}
Assume the mechanism $M=(A,P)$ is \emph{IR}.    
Then, for any agent $i$ and any bid profile $b \in V^n$, if $A(b) \neq i$, then $P_i(b) = 0$.

\item\label{lem:equalPayments}
Assume that $M$ is \emph{IC}. 
Then, for any agent $i$ and bid profiles $(b_i,b_{-i})$ and $(b'_i,b_{-i})$, if $A_i(b_i,b_{-i}) = A_i(b'_i,b_{-i})$, then $P_i(b_i,b_{-i}) = P_i(b'_i,b_{-i})$.    

\item \label{lem:MonoForSingleItem}
Assume that $M$ is \emph{IR} and \emph{IC}. 
Then, for any agent $i \in I$, bid profile $b = (b_i,b_{-i}) \in V^n$, and bid $b_i' > b_i$, if $\AM(b) = i$, then $\AM(b'_i,b_{-i}) = i$.

\end{enumerate}
\end{lemma}

Our notion of IC is ex-post: truthful bidding is a Nash equilibrium for every valuation profile. 
The standard properties listed in Lemma~\ref{lemma:simple properties} still follow under this notion (rather than IC in dominant strategies), since the deviations we consider are from truthful profiles with others bidding their true valuations.

\subsection{A notion of spitefulness}

The main theme of this study concerns the property of spitefulness, for which we provide a formal definition incorporated in the solution concept.
We say that a bid profile is a \emph{Spite-Free Nash equilibrium} if it maximizes every agent's payoff and there is no unilateral deviation that maintains the same payoff level for the deviating agent, while weakly reducing the payoffs of \emph{all} other agents and strictly reducing the payoffs of some.
This notion is formally given in Definition \ref{Definition - Three notions of IC} below.

\begin{definition} \label{Definition - Three notions of IC}
Fix a mechanism $M$ and a profile $v$ of private valuations. A bid profile $b$ is a \emph{Spite-Free Nash equilibrium (SNE)} if, for any agent $i$, there is no bid profile $b'=(b'_i,b_{-i})$ such that either $\UM_i(b';v_i) > \UM_i(b;v_i)$, or $\UM_i(b';v_i) = \UM_i(b;v_i)$ and $\UM_j(b';v_j) \le \UM_j(b;v_j)$ for all other agents $j\neq i$, with a strict inequality for at least one agent $j\neq i$.
\end{definition}

The SNE condition fits agents who have lexicographic preferences: they maximize their own utility, and, conditional on that, seek to reduce the payoffs of others.

As with the standard Nash equilibrium and other solution concepts, the SNE specifies a particular action profile but does not describe how agents might converge to it. 
In practice,  implementing an equilibrium in a model that involves incomplete information can be non-trivial, since an agent's payoff depends on her privately known type.
In our model, where the utility is linear in payment, increasing the winner's payment without affecting her winning probability is a spiteful behavior.
Such a behavior occurs naturally in combinatorial clock auctions,
where in the supplementary round an agent who bids her knock-out bid is guaranteed to obtain her final bundle in the dynamic phase.
Therefore, if all items have been allocated in the dynamic phase,
then posting additional bids in the supplementary round only leads to an increase in the payment of other bidders.

We next present the main solution concept that we introduce, which is robust to spiteful bidding.
A mechanism is spite-free incentive compatible if the profile where all agents bid truthfully is an SNE.

\begin{definition}
A mechanism is \emph{Spite-Free Incentive Compatible (SIC)} if, for every profile $v$, the bid profile $b=v$ of truthful bids is an \emph{SNE}.
\end{definition}

Second price sealed-bid auctions are not SIC.
Indeed, since the concept of SIC is an ex-post concept, an agent who did not obtain the item can increase its bid to be a bit less than the winning bid, thereby increasing the payment of the winner without changing the identity of the winner.

The symmetric equilibrium in first price sealed-bid auctions with symmetric bidders is not SIC as well.
Indeed, under this equilibrium, bidders usually bid strictly less than their private value.
Since the concept of SIC is ex post,
a bidder who did not win the item but whose private value lies above the winning bid has an incentive to increase her bid and obtain the item.

%%%%%%%%%%%%%%%%%%%

\section{Characterizing Spite-Free mechanisms} \label{sec:impossibility}

In this section, we present the paper’s main characterization result--Theorem~\ref{thm:charecterization}.  
The section is divided into two parts. 
In Section~\ref{sec:MainResult}, we define threshold mechanisms for a single indivisible good and characterize them as all IR and SIC mechanisms. 
In Section~\ref{section:impossibility}, we present two impossibility results (Corollaries~\ref{corollary:impossibility-anonymous} and \ref{corollary:SIC-Efficiency}) that build on Theorem~\ref{thm:charecterization}, showing that the only mechanism satisfying IR and SIC, and that is also either anonymous or efficient, is the null mechanism in which the item is never allocated.

\subsection{Main result} \label{sec:MainResult}

We here formally define \emph{threshold mechanisms} in Definition \ref{def:Threshold} below.
Roughly, a threshold mechanism consists of a sequence of take-it-or-leave-it offers, one for each agent, made according to a predetermined ordering of the agents and at exogenously fixed prices. The first agent who accepts her offer obtains the item at the stated price.

\begin{definition}\label{def:Threshold}
A mechanism $M=(A,P)$ is a \emph{threshold mechanism} if there exist a permutation $R: I \to \{1,\ldots,n\}$ on agents (i.e., a priority ranking), and a threshold $t_i \in \dR_+ \cup \{\infty\}$ for each agent $i \in I$, such that for any bid profile $b \in V^n$:
\begin{itemize}
\item If there exists an agent $j \in I$ such that $b_{j} > t_{j}$, then $A(b) = i$ and $P_{i}(b) = t_{i}$ where $i = \arg \min\{R(j): b_{j} \ge t_{j} \}$.
\item If $b_{j} \le t_{j}$ for every agent $j$,
then either $A(b) = j$ and $P_{j}(b) = t_{j}$ for some agent $j \in I$ satisfying $b_{j} = t_{j}$, or $A(b) = 0$.  
%\item If $b_{j} < t_{j}$ for every agent $j$, then $A(b) = 0$.
\end{itemize}
\end{definition}

The permutation $R$ in Definition~\ref{def:Threshold} specifies a priority ranking of the agents, with each agent assigned a threshold value for obtaining the item.  
If no agent bids above her assigned threshold, all agents receive zero ex-post payoffs, as in the case where the item remains unallocated.  
Even when the priority ranking and thresholds are fixed, the threshold mechanism is not unique, since variations in the allocation rule may arise when no agent bids strictly above her threshold.

The next result asserts that a mechanism is IR and SIC if and only if it is a threshold mechanism.  
To build some intuition, consider the second-price mechanism which we discussed earlier.
This mechanism is IR and IC, yet not spite-free: any non-winning agent can increase her bid to reduce the payoff of the winning agent. 
Such a behavior is impossible under threshold mechanisms, since the priority ranking and payments are predetermined.  
This observation captures the essence of the spitefulness property: the price paid by the winning agent cannot depend on the actions of non-winning agents, as such dependence would allow them to act spitefully.

\begin{theorem} \label{thm:charecterization}
A mechanism is \emph{IR} and \emph{SIC} if and only if it is a threshold mechanism. 
\end{theorem}

The proof of Theorem~\ref{thm:charecterization} is extensive and is therefore deferred to Appendix~\ref{sec:ExtremeSpiteAppendix}.  
It is based on a stronger notion of spitefulness, referred to as \emph{Extreme SIC} (ESIC), in which agents not only bid truthfully to maximize their own payoffs in equilibrium, but also satisfy the following stricter (equilibrium) condition: there exists no unilateral deviation that leaves the deviator’s payoff unchanged while \emph{strictly} reducing the payoff of at least one other agent (even if such a deviation increases the payoffs of others).  
Clearly, every ESIC mechanism is also a SIC mechanism. 
We will show that an IR mechanism is ESIC if and only if it is a threshold mechanism, and then extend the result to SIC mechanisms.
In particular, our proof implies that all SIC mechanisms are ESIC.

\subsection{Efficient and Anonymous Spite-Free mechanisms}
\label{section:impossibility}

There are two immediate impossibility results that follow from Theorem \ref{thm:charecterization}.
These results state that the only IR and SIC mechanism that sustains an additional condition of either efficiency, or anonymity, is the \emph{null mechanism}, where, for each bid profile, all payments are zero and the item is never allocated to any agent.
Formally, a mechanism $M$ is a null mechanism if $\AM(b) = 0$ and $\PM_i(b) = 0$ for each agent $i$ and every bid profile $b$.
Alternatively, a null mechanism is a threshold mechanism where $t_i =\infty$ for every agent $i$.

Under the (threshold) null mechanism,
the utility of each agent is always zero.
Therefore, this mechanism is anonymous.
Since under the null mechanism the item is never allocated, this mechanism is also efficient. 
\color{black}
The first impossibility result, stated in Corollary~\ref{corollary:impossibility-anonymous}, asserts that any anonymous, IR, and SIC mechanism must be a null mechanism.  
Given the characterization of IR and SIC mechanisms as threshold mechanisms, the intuition is immediate: once a priority ranking is fixed with some feasible payments, the mechanism can no longer be anonymous.

\begin{corollary} \label{corollary:impossibility-anonymous}
If a mechanism is \emph{IR}, Anonymous, and \emph{SIC}, then it is the null mechanism. 
\end{corollary}

The second impossibility result, stated in Corollary~\ref{corollary:SIC-Efficiency}, asserts that any efficient, IR, and SIC mechanism must also be a null mechanism.  
The intuition again follows directly from the characterization of IR and SIC mechanisms as threshold mechanisms: once (finite) thresholds and a priority ranking are fixed, the outcome need not be efficient, since agents with higher valuations may be ranked below agents with lower valuation, implying that the winning agent does not necessarily submit the highest bid.

\begin{corollary} \label{corollary:SIC-Efficiency}
If a mechanism is \emph{IR}, efficient and \emph{SIC}, then it is the null mechanism.     
\end{corollary}

Both Corollaries~\ref{corollary:impossibility-anonymous} and \ref{corollary:SIC-Efficiency} follow directly from Theorem~\ref{thm:charecterization},  and hence their proofs are omitted.

\section{Extensions and Open problems}

\subsection{Efficiency of the threshold mechanism} \label{Section: optimal spite free mechanism - example}

A natural question regards the efficiency loss due to the SIC requirement. 
To properly compare the efficiency of optimal auctions to those of threshold mechanisms,
we need to assume that the agents' private values are distributed according to some distribution.
As we now argue, when the agents are symmetric and their private values are independent, as the number of bidders tend to infinity, there is no efficiency loss.
Indeed, suppose that each agent's private value is distributed according to some CDF $F$.
Assume w.l.o.g.~that the upper end of the support of $F$ is $1$.
Consider the threshold mechanism in which the thresholds of \emph{all} agents are $1-\ep$, for some $\ep > 0$.
As $n$ goes to infinity, the probability that at least one agent's private value exceeds $1-\ep$ goes to $1$,
and hence asymptotically the seller's gain is at least $1-\ep$.
Since $\ep$ is arbitrary, the efficiency loss is asymptotically negligible. 

In fact, when $F$ is the uniform distribution,
if the ranking of the agents is given by the decreasing order where agent $n$ is the first, and agent $1$ is the last in line,
then the optimal thresholds $(t_i)_{i=1}^n$ satisfy the following recursion:
\[ t_1 = \frac{1}{2}, \ \ \ t_{i+1} = \frac{1+(t_{i})^2}{2}, \forall i\geq 1, \]
which converges to $1$ as $n$ increases (see analysis in Appendix \ref{Appendix: optimal spite free mechanism - example}). 
In addition, the seller's expected revenue under these thresholds, denoted $\gamma_n$, satisfies the recursive equation
\[ \gamma_1 = \frac{1}{4}, \ \ \ \gamma_{n+1} = (1-t_{n+1})t_{n+1} + t_{n+1}\gamma_{n}, \ \ \ \forall n\geq 1, \]
which converges to $1$ as well.

Yet, an open and interesting question for future research is the efficiency loss due to the SIC requirement for a given $n$ (rather than asymptotically), and the efficiency loss for heterogeneous agents.

\subsection{Extension to a multi-item setting} \label{sec:Multiple}

An applicative extension of the SIC requirement is to environments with multiple items. 
Once agents can choose bundles, 
whenever an agent is indifferent between several bundles, by misreporting her valuations she may be able to determine which of those bundles she would obtain.
This way, the agent can change the residual supply, thereby producing sharp changes in others' opportunities, making general positive results significantly harder to obtain. 

Nevertheless, there exists a broad class of multi-item mechanisms that are spite-free: mechanisms that eliminate cross-agent externalities by design.
One possibility is to partition the market into agent-separable sub-markets, so that each agent optimizes over a distinct subset of items that affects only her own allocation and payment. Since an agent's report cannot change anyone else's feasible set or prices, deviations cannot harm others, and truthful reporting maximizes the true objective given that sub-market. Another possibility is to treat all items as a single indivisible bundle, reducing the multi-item problem to a single-item one. These mechanisms, however, are limited and need not generate desirable allocations.

In this section we will extend the basic model to a multi-item setting and discuss threshold mechanisms within this framework.

\subsubsection{The updated model}

In the multiple items setting, we study the problem of assigning a finite set $S = \{a_1,a_2,\ldots,a_K\}$ of items to the set of agents $I = \{1,2,\ldots,n\}$. 
The set of all \emph{bundles} is $2^S$.
For each agent $i \in I$, the valuation function $v_i: 2^S \to \dR_+$ determines the worth agent $i$ assigns to each bundle, with $v_i(\emptyset) = 0$ for each $i \in I$.
The set of valuation functions is denoted by 
$V := \{ v \in \dR_+^{2^S} \colon v(\emptyset) = 0 \}$. 

A \emph{bid} $b_i \in V$ of agent $i$ is the reported valuation of agent $i$. 
A \emph{bid profile} $b = (b_1,\ldots,b_n) \in V^n$ is a vector of bids, one for each agent. 
%Throughout the paper, we also denote the bid profile $b$ by $(b_i,b_{-i})$ to emphasis agent $i$.
An \emph{allocation} is an assignment of bundles for each agent that satisfies the feasibility constraint 
$\bigcup_{i \neq j} (T_i \cap T_j) = \emptyset$, where $T_i$ and $T_j$ are the assigned bundles of agents $i$ and $j$, respectively. 
A \emph{payment} for each agent is a non-negative real number. 
The mechanism, the utility of agents, and the SIC property are defined similarly to the analogous concepts in the single item setup.

\subsubsection{Two extensions of threshold mechanisms}
\label{sec:examples}

We here present two natural extensions of threshold mechanisms, 
called \emph{cluster mechanisms} and \emph{sequential mechanisms}.

In a cluster mechanism, agents are ranked, 
and each agent $i$ has a threshold $t_i^T$ for each bundle $T$.
In her turn,
each agent $i$
is allocated a bundle (out of the items that have not been allocated so far) which maximizes her gain according to her bid:
$T_i \in \argmax_{T \subseteq S \setminus \bigcup_{j < i}T_{j}} \left\{\, b_i(T)-t_i^T\, \right\}$,
where $T_j$ is the bundle allocated to agent $j$, for each $j < i$, so that $S \setminus \bigcup_{j < i}T_{j}$ is the set of remaining items for agent~$i$ and the following agents.
If there are several bundles that attain the maximum for agent~$i$,
some tie-breaking rule dictates which among them will be allocated to the agent.

In a sequential mechanism,
each agent has a threshold $t_i$,
and the items are allocated one after the other:
Item $a_j$ is allocated to the 
highest-ranked agent $i$ that satisfies that the item's marginal contribution to the set of items already allocated to $i$ is at least her threshold.

These mechanisms are not SIC.
In fact, sequential mechanisms are not even IC.
Indeed, suppose there are two items, a single agent, and thresholds are $t_1,t_2$.
Suppose the agent values one item at $0$ and two items at $t_1+t_2+\ep$ for some $\ep > 0$.
By submitting her true valuation the agent is not allocated any item and gains $0$, while by submitting a bid $b_1$ such that $b_1(\{a_1\}) = t_1 + \ep/2$ and $b_1(\{a_1,a_2\})=t_1+t_2+\ep$, the agent is allocated both items and gains $\ep$.%
\footnote{The failure described above is due to the fact that the mechanism is sequential.
Another failure of IC in sequential mechanisms arises since the items are allocated in a given order.
Consider the example above, where the agent's
valuation is $v_{1}(\{a_1\})=t_1$,
$v_{1}(\{a_2\}) = t_2+\ep$ for some $\ep > 0$,
and $v_{1}(\{a_1,a_2\}) = 0$.
A sequential mechanism will allocate $a_1$ to the agent,
generating a gain of $0$,
while if the agent submits the bid $b_{1}(\{a_1\}) = b_{1}(\{a_1,a_2\}) = 0$ and $b_{1}(\{a_2\}) = t_2+\ep$,
she will be allocated $a_2$ and gain $\ep$.
}

A more delicate weakness of these two mechanisms arises when agents are indifferent between several bundles, 
and by submitting a bid different from their true valuation,
can change the bundle they are allocated,
and hence also the bundle subsequent agents are allocated.
Consider, for example, 
the case of two items and two bidders.
Consider a cluster mechanism with agent-independent thresholds $t^{\{1\}}$, $t^{\{2\}}$, and $t^{\{1,2\}}$
and suppose that when the bid of agent~$1$ is
\begin{equation}
\label{equ:b}
b_1(\{a_1\}) = t^{\{1\}}, \ \ \ b_1(\{a_2\}) = t^{\{2\}}, \ \ \ b_1(\{a_1,a_2\}) = 0, 
\end{equation} 
the tie-breaking rule allocates $\{a_1\}$ to agent~$1$.
Suppose that agent~$1$'s valuation coincides with the bid that appears in Eq.~\eqref{equ:b},
and that agent~$2$'s valuation is
\[ v_2(\{a_1\}) = t^{\{1\}}, \ \ \ v_2(\{a_2\}) = t^{\{2\}}+\ep, \ \ \ v_2(\{a_1,a_2\}) = 0, \]
for some $\ep > 0$.
The cluster mechanism will allocate $\{a_1\}$ to agent~$1$
and $\{a_2\}$ to agent $2$,
yielding them gains of $0$ and $\ep$, respectively.
However,
if agent~$1$ spitefully bids
\[ b_1(\{a_1\}) = 0, \ \ \ b_1(\{a_2\}) = t_2, \ \ \ b_1(\{a_1,a_2\}) = 0, \]
the allocation will be $\{a_2\}$ to agent~$1$ and $\{a_1\}$ to agent~$2$,
yielding both a gain of $0$.
Thus, by misreporting, agent~$1$ lowers agent~$2$'s gain without lowering her gain.
A similar example can be constructed for the sequential mechanism.

\subsubsection{Homogeneous items and submodular valuations}

The examples provided in Section~\ref{sec:examples} show that in the multi-item setup,
general positive results are difficult to obtain. 
To eliminate equal-payoff bundles that are the source of this difficulty, we will restrict ourselves here to the case where items are homogeneous and the agents' valuations are submodular.
In that case, the value an agent assigns to each bundle depends only on the number of items in the bundle: $v_i(T_i) = \widehat v_i(|T_i|)$, for some function $\widehat v_i : \dN \to \dR_+$.
Moreover, the function $\widehat v_i$ is submodular, namely, the marginal contribution of each item is weakly decreasing with the number of items the agent is already allocated:
the function $k \mapsto (\widehat v_i(k+1) - \widehat v_i(k))$ is non-increasing.
We call such valuations \emph{homogeneous and submodular}.

\begin{theorem} \label{Theorem - sequential mechanism}
Consider the multi-item setup, when agents valuations and bids are restricted to homogeneous and submodular valuations.
Then, sequential mechanisms are \emph{IR}, \emph{IC}, and \emph{SIC}.
\end{theorem}

The proof is given in Appendix \ref{Appendix - proof of Theorem for sequential}.
Let us explain the intuition for this result with the highest ranked agent, agent~$1$.
Consider the sequential mechanism, and suppose that by truly reporting her valuation, agent~$1$ is allocated $r$ items.
In particular, the marginal contribution of the $(r+1)$'st item to agent~$1$ is lower than her threshold.
Since valuations are homogeneous and submodular,
the marginal contribution of more than one item to agent~$1$'s current bundle is lower than her threshold as well.
Hence, agent~$1$ does not profit by misreporting in a way that 
enlarges her allocated bundles.
If agent~$1$ misreports her valuation in a way that her allocated bundle becomes smaller,
then more items are left to the other agents.
Suppose that under truthful reporting agent~$2$ obtained $s$ items.
If $r+s = K$, agent~$2$ may now get a larger bundle and profit, while subsequent agents cannot lose, because they got the empty bundle under truthful reporting.
If $r+s < K$, then agent~$2$ does not want more than $s$ items, and hence the excess that is caused by the misreport of agent~$1$  leads to an excess supply to agent~$3$ (and possibly to subsequent agents), and by induction they cannot lose.

In threshold mechanisms, the threshold for a bundle is the size of the bundle times the threshold of a single item.
Therefore, when valuations are submodular, the difference between the valuation of the bundle and its threshold is submodular as well.
In cluster mechanisms, without imposing further restrictions on the thresholds, this difference need not be submodular, and hence the mechanism is not necessarily IC even when valuations are submodular.
To overcome this difficulty, one option is to restrict attention to cluster mechanisms where the thresholds are strictly supermodular.
This condition eliminates size-based indifferences and ensures that agents' incentives are aligned across all relevant bundle choices.

\subsubsection{Necessary conditions for SIC}

Although we do not have a general characterization of SIC mechanisms in the multi-item setup,
we can derive necessary conditions for a mechanism to be IR and SIC in a general valuation domain.  

The following lemma, which serves as the multi-item analog of Lemma~\ref{lemma:simple properties}(\ref{lem:equalPayments}), establishes that the payment of agent~$i$ depends only on $T_i$ and $b_{-i}$.
Intuitively, once an agent’s allocation is determined, any further dependence of her payment on her report would generate profitable deviations, violating incentive compatibility.  
%The following lemma formalizes this property and serves as a foundational building block for the results that follow.

\begin{lemma}\label{lem:payment}
Assume that the mechanism $M$ is \emph{SIC}. 
For any agent $i$ and any bid profiles $(b_i,b_{-i})$ and $(b'_i,b_{-i})$, if $A_i(b_i,b_{-i}) = A_i(b'_i,b_{-i})$, then $P_i(b_i,b_{-i}) = P_i(b'_i,b_{-i})$.    
\end{lemma}

\begin{proof}

Assume to the contrary that $A_i(b_i,b_{-i}) = A_i(b'_i,b_{-i}) = T_i \subseteq S$ and $P_i(b_i,b_{-i}) > P_i(b'_i,b_{-i})$. 
Then we have 
\[
u_i(b_i,b_{-i};b_i) = b_i(T_i) - P_i(b_i,b_{-i}) < b_i(T_i) - P_i(b'_i,b_{-i}) = u_i\big(b'_i,b_{-i};b_i\big).
\]
Hence, $b'_i$ is a profitable deviation at $(b_i,b_{-i})$, contradiction to the assumption that $M$ is SIC. 
\end{proof}

%\ES{I don't understand the next sentence. What will change in the writing if we don't make this convention?}

%We adopt the convention that all comparisons in this subsection are evaluated at truthful bids $b=v$.

For $T_i\subseteq S$, let 
$$
W_{-i}^{T_i} \;:=\; \big\{\,b_{-i}\in V_{-i}\ :\ \exists\,b_i\in V_i\ \text{with}\ A_i(b_i,b_{-i})=T_i\,\big\}$$ 
be the set of all $b_{-i} \in W_{-i}^{T_i}$ for which there exists $b_i \in V_i$ with $A_i(b_i,b_{-i}) = T_i$. 
By Lemma~\ref{lem:payment}, the payment to agent $i$ conditional on receiving $T_i$ depends only on $(T_i,b_{-i})$, so we denote it by
$P_i^{T_i}:W_{-i}^{T_i}\to\dR_+$, and extend it to all $V_{-i}$ by setting $P_i^{T_i}(b_{-i})=\infty$ when $b_{-i}\notin W_{-i}^{T_i}$.

The next result provides necessary conditions for a mechanism to be both IR and SIC.
It shows that allocations must coincide with the agent’s utility-maximizing choice under the corresponding payment schedule, and it also bounds the richness of feasible payment profiles in the two-agent case.
The proof is deferred to Appendix \ref{Appendix - Proof of Theorem 2}.

\begin{theorem}\label{thm:MultipleChar}
Assume $M$ is \emph{IR} and \emph{SIC}.
Fix agent $i$, a non-empty $T_i\subseteq S$, and $b=(b_i,b_{-i})\in V$. Then:
\begin{enumerate}
    \item If $b_i(T_i)-P_i^{T_i}(b_{-i}) > b_i(S_i)-P_i^{S_i}(b_{-i})$ for all $S_i\subseteq S$ with $S_i\neq T_i$, then $A_i(b)=T_i$.
    \item If $A_i(b)=T_i$, then $b_i(T_i)-P_i^{T_i}(b_{-i}) \ge b_i(S_i)-P_i^{S_i}(b_{-i})$ for all $S_i\subseteq S$ with $S_i\neq T_i$.
    \item If there are only two agents, then the range of $P_i^{T_i}: W^{T_i}_{-i}\to\dR_+$ has cardinality at most $2^{K-|T_i|}$.
\end{enumerate}
\end{theorem}

Theorem~\ref{thm:MultipleChar} establishes the structure that any IR and SIC mechanism must satisfy.  
Parts~(1) and~(2) are natural extensions of Theorem~\ref{thm:charecterization}.
Together they imply that the allocation rule must coincide with the agent’s utility-maximizing bundle under the induced payment schedule, effectively embedding a choice-theoretic consistency condition into the mechanism.  
Part~(3) adds a combinatorial restriction: in two-agent environments, the number of distinct attainable payments is bounded by the number of possible residual allocations of the remaining items.  
This bound reflects the tight informational linkage between agents’ allocations and payments, a feature that highlights the limited degrees of freedom available in designing multi-item SIC mechanisms.

Parts~(1) and~(2) of Theorem~\ref{thm:MultipleChar} provide a geometric interpretation of the allocation rule.  
For each agent~$i$, every possible allocation $T_i \subseteq S$, and every fixed profile of other agents’ bids $b_{-i}\in V_{-i}$, define the set
\[
V_i^{T_i}(b_{-i}) = 
\Big\{\,b_i \in V_i : b_i(T_i) - P_i^{T_i}(b_{-i}) > b_i(S_i) - P_i^{S_i}(b_{-i})
\ \text{for all}\ S_i \subseteq S,\ S_i \neq T_i \Big\}.
\]
Theorem~\ref{thm:MultipleChar} implies that for any $b_i \in V_i^{T_i}(b_{-i})$, we have $A_i(b_i,b_{-i}) = T_i$.  
Thus, for a fixed $b_{-i}$, the sets $\{V_i^{T_i}(b_{-i})\}_{T_i\subseteq S}$ partition the space of bids $V_i$ into regions in which each allocation is strictly optimal.  
Each such region is a convex polytope, and these polytopes are pairwise disjoint except possibly along their boundaries.  
The theorem therefore determines the allocation for all bids lying in the interior of one such polytope.  
However, on the boundaries, the allocation may correspond to any of the sets whose regions meet at that point, as Part~(2) indicates.  
We illustrate this with the following example.

\begin{example}[two items, two agents, additive bids]\rm 
Let $S=\{a_1,a_2\}$,
and let $M = (A,P)$ be an IR and SIC mechanism.
By Theorem~\ref{thm:MultipleChar},
agent~$1$'s payment depends on the bundle she gets and on agent~$2$'s bid.
Suppose that for a fixed bid $b_2$ of agent~$2$,
\[
P_1^{\{a_1\}}(b_2)=4,\qquad P_1^{\{a_2\}}(b_2)=3,\qquad P_1^{\{a_1,a_2\}}(b_2)=6,\qquad P_1^{\emptyset}(b_2)=0.
\]
Suppose moreover that 
agent $1$’s valuation is additive: $v_i(\{a_1\})=x,\ v_i(\{a_2\})=y,\ v_i(\{a_1,a_2\})=x+y$ with $x,y\ge 0$.  
For each bundle $T\subseteq S$, 
agent~$1$'s utility from obtaining bundle $T$, denoted $U_1^T = v_i(T)-P_i^{T}$, is, then,
\[
U_1^{\{a_1\}}=x-4,\qquad U_1^{\{a_2\}}=y-3,\qquad U_1^{\{a_1,a_2\}}=x+y-6,\qquad U_{\emptyset}=0.
\]

By Parts~(1)–(2) of Theorem~\ref{thm:MultipleChar}, the region 
$V_1^T(b_2)$ in which $A_1(b_1,b_2)=T$ is the set of $(x,y)$ where $U_1^T$ strictly exceeds all other $U_1^{T'}$. 
These regions are intersections of half-spaces (hence convex polyhedrons), meeting on shared boundaries where ties occur.

\begin{itemize}
    \item So that $A_1(b_1,b_2) = \{a_1\}$ we need $U_1^{\{a_1\}}>\max\{U_1^{\{a_2\}},U_1^{\{a_1,a_2\}},U_1^{\emptyset}\}$.
    Hence, $V_1^{\{a_1\}}(b_2)=\{(x,y):\ x>4,\ y<2,\ x-y>1\}$.
    \item  So that $A_1(b_1,b_2) = \{a_2\}$ we need $U_1^{\{a_2\}}>\max\{U_1^{\{a_1\}},U_1^{\{a_1,b_2\}},U_1^{\emptyset}$\}.
    Hence, $V_1^{\{a_2\}}(b_2)=\{(x,y):\ y>3,\ x<3,\ y>x-1\}$.
    \item So that $A_1(b_1,b_2) = \{a_1,a_2\}$ we need $U_1^{\{a_1,a_2\}}>\max\{U_1^{\{a_1\}},U_1^{\{a_2\}},U_1^{\emptyset}\}$.
    Hence, $V_1^{\{a_1,a_2\}}(b_2)=\{(x,y):\ x>3,\ y>2,\ x+y>6\}$.
\end{itemize}

The boundaries of these regions represent bid profiles where there are ties: the lines $y=2$ (between $\{a\}$ and $\{a,b\}$), $x=3$ (between $\{b\}$ and $\{a,b\}$), and $x+y=6$ (between $\{a,b\}$ and $\emptyset$),
see Figure \ref{fig:regions_ab}.
By Part~(2), on the boundary, the allocation can be any of the tied sets. 
Outside these regions, $\emptyset$ is weakly optimal ($U_{\emptyset}=0$), consistent with IR.
\end{example}

\begin{figure}[th!]
    \centering
\begin{tikzpicture}
  \begin{axis}[
    width=9cm,height=9cm,
    xmin=-0.4,xmax=8,ymin=-0.4,ymax=8,
    axis lines=middle,
    axis line style={->,line width=1.2pt},   % bolder axes with arrows
    tick style={line width=0.9pt},
    major tick length=3pt,
    xlabel={$x$}, ylabel={$y$},
    xlabel style={at={(axis description cs:1,0.03)},anchor=north east, font=\small},
    ylabel style={at={(axis description cs:0.03,1)},anchor=south west, font=\small},
    xtick={0,1,2,3,4,5,6,7}, ytick={0,1,2,3,4,5,6,7},
    ticklabel style={font=\small},
    clip=false
  ]

    % --------- Fill the three sets (colors only) ----------
    % V_i^{ {a} } : x>4, y<2
    \path[fill=blue!60,opacity=.40]  (axis cs:4,0) rectangle (axis cs:8,2);

    % V_i^{ {b} } : y>3, x<3
    \path[fill=red!60,opacity=.40]   (axis cs:0,3) rectangle (axis cs:3,8);

    % V_i^{ {a,b} } : x>3, y>2, x+y>6
    % (i) rectangle part for x>=4
    \path[fill=green!60,opacity=.40] (axis cs:4,2) rectangle (axis cs:8,8);
    % (ii) wedge above y=6-x for 3<=x<=4
    \path[fill=green!60,opacity=.40]
      (axis cs:3,8) -- (axis cs:3,3) -- (axis cs:4,2) -- (axis cs:4,8) -- cycle;

    % --------- Region labels ----------
    \node[blue!75!black]  at (axis cs:6,0.7) {$V_1^{\{a_1\}}(b_2)$};
    \node[red!75!black]   at (axis cs:1.4,5.2) {$V_1^{\{a_2\}}(b_2)$};
    \node[green!60!black] at (axis cs:6,5.0) {$V_1^{\{a_1,a_2\}}(b_2)$};
    \node at (axis cs:1.4,0.7) {$V_1^{\emptyset}(b_2)$};

    % --------- Dashed boundaries ----------
    \draw[dashed, thick] (axis cs:4,0) -- (axis cs:4,2);   % x=4
    \draw[dashed, thick] (axis cs:3,0) -- (axis cs:3,8);   % x=3
    \draw[dashed, thick] (axis cs:0,2) -- (axis cs:8,2);   % y=2
    \draw[dashed, thick] (axis cs:0,3) -- (axis cs:3,3);   % y=3
    \addplot[black, dashed, thick, domain=3:4, samples=2] {6 - x}; % x+y=6

  \end{axis}
\end{tikzpicture}
    \caption{The partition of the valuation space into the sets $(V_1^T(b_2))_{T \subseteq \{a_1,a_2\}}$ for a given $b_2$.}
    \label{fig:regions_ab}
\end{figure}

\subsection{Summary}

Spiteful behavior, while not common,  does occur in practice.
The literature studied spiteful behavior by altering the utility functions of the agents: an agent utility function is the weighted difference between the gain of the agent and the gain of her competitors.
This approach implicitly assumes that an agent is willing to bear a loss if this will cause her competitors to suffer a much greater loss.

However, shareholders are usually not willing to suffer a loss just to inflict a larger loss on their competitors.
This makes the agents' utility lexicographic: agents would like to maximize their gain, and contingent upon their own gain, they may seek to lower their competitors' gains.
As mentioned in Section~\ref{Section - the model}, the supplementary phase of combinatorial clock auctions gives bidders an opportunity to exercise such behavior.

In the single-item setup, threshold mechanisms are the only IR and SIC mechanisms.
We hope that our study will lead to the complete characterization of IR and SIC mechanisms in the multi-item setup, as well as to the identification of the optimal one.

\bibliographystyle{aer}
\bibliography{./references.bib}

\newpage

\appendix

\begin{appendices}

\section{A notion of Extreme Spitefulness and the proof of Theorem \ref{thm:charecterization}} \label{sec:ExtremeSpiteAppendix}

To prove our main result, we require an additional spitefulness notion, referred to as \emph{Extreme Spite-Free Nash equilibrium} (ESNE) which allows agents to deviate if the deviating agent maintains the same utility level, while strictly reducing the payoffs of some other agents, but not necessarily all of them.

\begin{definition}
A profile $b$ is an \emph{Extreme Spite-Free Nash equilibrium} \emph{(ESNE)} if, for any agent $i$, there is no bid profile $b'=(b'_i,b_{-i})$ such that either $\UM_i(b';b_i) > \UM_i(b;b_i)$, or $\UM_i(b';b_i) = \UM_i(b;b_i)$ and $\UM_j(b';b_j) < \UM_j(b;b_j)$ for at least one agent $j\neq i$.
A mechanism is \emph{Extreme Spite-Free Incentive Compatible (ESIC)} if the profile $b=v$ of truthful bids is an \emph{ESNE}.
\end{definition}

Notice that an ESIC mechanism is also a SIC mechanism, and in case of two agents ($n=2$), the two spite-free notions of SIC and ESIC coincide.

The proof of Theorem~\ref{thm:charecterization} consists of three steps:
\begin{itemize}
\item Every threshold mechanism is IR and ESIC (and therefore also SIC) (Lemma~\ref{lemma:IR}).
\item Every IR and ESIC mechanism is a threshold mechanism (Theorem~\ref{thm:ESIC-Char}).
\item Every IR and SIC mechanism is ESIC (Theorem~\ref{thm:SIC}).
\end{itemize}

The proof yields a stronger property of IR and SIC mechanisms:
they are in fact ESIC.

An important concept we will need is that of single perturbation paths,
which is
a sequence of bid profiles $(b^r)_{r=0,1,\ldots,n}$ such that
$b^r$ and $b^{r+1}$ coincide in all coordinate except possibly the bid of agent $r+1$, for each $0 \leq r \leq n-1$.
Formally,

\begin{definition}
A sequence of bid profiles $(b^r)_{r=0,1,\ldots,n}$ is called a \emph{single perturbation path} if for each $r=0,1,\ldots,n-1$, we have $b^r_k = b^{r+1}_k$ for each $k \neq r+1$.  
\end{definition}

\subsection{Threshold mechanisms are IR and ESIC}

In this section we prove the following result.

\begin{lemma}
\label{lemma:IR}
Every threshold mechanism is \emph{IR} and \emph{ESIC}.
\end{lemma}

\begin{proof}
Let $M$ be a threshold mechanism with permutation $R$ and thresholds $(t_i)_{i \in I}$.

$M$ is IR because the item can be allocated only to an agent~$i$ whose bid is at least $t_i$,
and if agent~$i$ obtains the item, she pays $t_i$.

We next prove that $M$ is IC.
Fix a profile of private values $v$,
and suppose the agents bid truthfully, that is, $b=v$.
If $A(b) = i$,
then necessarily $v_i \geq t_i$,
and agent~$i$'s utility is $v_i - t_i \geq 0$.
By deviating,
either agent~$i$ still get the item and pays $t_i$,
or does not get the item.
In both cases agent~$i$ does not profit.

If $A(b) \neq i$,
then agent~$i$'s utility is $0$.
If agent~$i$ has a profitable deviation,
then necessarily $v_i > t_i$.
But since $A(b) \neq i$,
this means that $A(b) = j \neq i$
and $j$ is ranked higher than $i$,
so that no deviation of $i$ can change the identity of the winner, a contradiction.

We finally prove that $b=v$ is an SNE.
Suppose by contradiction that agent~$i$ has a deviation that harms agent~$j$.
Since $M$ is IR, 
agent $j$ can be harmed only if
at $b=v$ agent~$j$ wins the item, and her payoff is positive.
Agent~$i$ can affect the identity of the winner only if she is ranked higher than $j$.
Since $A(v) = j$,
necessarily $v_i < t_i$.
Therefore a deviation of agent~$i$ that affects the identity of the winner is to a bid $b'_i \geq t_i$,
in which case agent~$i$ will obtain the item and her utility will be $v_i - t_i < 0$.
\end{proof}

\subsection{Winner's payment in ESIC mechanisms}

In this section we
show that, given an IR and ESIC mechanism, and conditional on a winning agent $i$, the payment of agent $i$ is independent of the bid profile.
This is the analog of Lemma~\ref{lemma:simple properties}(\ref{lem:equalPayments}) in the multi-item setup.

\begin{lemma}\label{lem:stairESIC}
Fix an \emph{IR} and \emph{ESIC} mechanism $M$.
If $A(b) = A(b') = i$ for some agent $i \in I$ and bid profiles $b,b' \in V^n$, then $P_i(b) = P_i(b')$.
\end{lemma}

\begin{proof}
The proof comprises three steps.

\noindent\textbf{Step 1:} Definitions.

Assume to the contrary that there exists an agent $i \in I$ and bid profiles $b,b' \in V^n$, such that $A(b) = A(b') = i$ and $P_i(b) > P_i(b')$.
Define two bid profiles $c = (c_i,b_{-i})$ and $c'=(c_i',b'_{-i})$, where $c_i=c_i' = \max\{b_i,b_i'\} + \epsilon$ for some $\epsilon>0$.
By Lemma \ref{lemma:simple properties}(\ref{lem:MonoForSingleItem}), $A(c)=A(c')=i$, and by Lemma~\ref{lemma:simple properties}(\ref{lem:equalPayments}),
$\PM_i(c) = \PM_i(b)$ and $\PM_i(c') = \PM_i(b')$.
Therefore, $P_i(c) > P_i(c')$, and hence $u_i(c';c_i) > u_i(c;c_i)$.

Consider the single perturbation path $(c^r)_{r=0,\ldots,n}$, where $c^0 = c$ and $c^n = c'$.
Note that $P_i(c^0) = P_i(c) > P_i(c') = P_i(c^n)$,
and since $A(c^0) = A(c^n) = i$, 
we have $\UM_j(c^0;b_j) = \UM_j(c^n;b'_j) = 0$ for each agent $j \neq i$. 

Define a set of agents $I_1 \subseteq I$ as $I_1 = \{r: \UM(c^{r-1}) \neq \UM(c^{r})\}$, where $\UM(b) = (u_j(b;b_j))_{j \in I}  \in \dR^n$ is the utility vector when each agent's bid in the bid profile $b$ matches her true valuation, and a set $I_2 \subseteq I$ as $I_2 = \{j: \exists r \ {\text{such that}} \ \UM_j(c^{r};c^r_j) \neq \UM_j(c^{r+1};c^{r+1}_j)\}$.
Agent $i$'s bid is constant throughout the single perturbation path, and therefore $c^{i-1} = c^i$, which implies that $i \notin I_1$.
Since $u_i(c';c_i) > u_i(c;c_i)$, $i \in I_2$.

\bigskip
\noindent\textbf{Step 2:}
$|I_1| \ge |I_2|$.

Define a multi-valued function $F:I_1 \to I_2$ as follows.
For every $r \in I_1$ and every $j \in I_2$, 
set $j \in F(r)$ if and only if $\UM_{j}(c^{r-1};c^{r-1}_j) \neq \UM_{j}(c^{r};c^{r}_j)$. 
By Lemma~\ref{lemma:simple properties}(\ref{lem:ZeroPayment}), at most one agent has a strictly positive utility in any bid profile, 
and hence $|F(r)| \le 2$ for every $r \in I_1$. 
For any $j \in I_2$ such that $j \neq i$, we have $\UM_j(c^{0};b_j) = \UM_j(c^{n};b_j') = 0$ and $\UM_j(c^{r};c^r_j) > 0$ for some $1 \le r < n$, hence $|F^{-1}(j)| \ge 2$. 
In addition, the condition $\UM_i(c^{0};c_i) < \UM_i(c^{n};c_i)$ implies that $|F^{-1}(i)| \ge 1$.
Therefore,
$$
2|I_1| \geq \sum_{r\in I_1}|F(r)| = \sum_{j\in I_2} |F^{-1}(j)| \geq 1 + 2(|I_2|-1) = 2|I_2|-1,
$$
where the first inequality holds since $|F(r)| \le 2$ for every $r \in I_1$,
and the second inequality holds since $|F^{-1}(i)| \ge 1$
and $|F^{-1}(j)| \ge 2$ for every $j \in I_2 \setminus \{i\}$.
Thus, $|I_1| \ge |I_2|$, as claimed.

\bigskip
\noindent\textbf{Step 3:} Deriving a contradiction.

Because $i \notin I_1$ and $i \in I_2$, there exists an agent $j$ such that $j \in I_1$ and $j \notin I_2$.
Since $j \notin I_2$, we have $\UM_{j}(c^{j-1};c^{j-1}_j) = \UM_{j}(c^{j};c^{j}_j)$,
and since $j \in I_1$, we have $\UM_{j'}(c^{j-1};c^{j-1}_{j'}) \neq \UM_{j'}(c^{j};c^{j}_{j'})$ for some agent $j'$.
Hence, at either $c$ or $c'$, agent $j$ can misreport her bid to decrease the utility of agent $j'$ while keeping her own utility fixed.
This contradicts the ESIC property and concludes the proof.
\end{proof}

\subsection{Characterization of IR and ESIC mechanisms}

In this subsection, we characterize the set of IR and ESIC mechanisms.

\begin{theorem}\label{thm:ESIC-Char}
Every \emph{IR} and \emph{ESIC} mechanism is a threshold mechanism. 
\end{theorem}

\begin{proof}
Let $M = (A,P)$ be an IR and ESIC mechanism.
For each agent $i$, define a function $T_i: V^{n-1} \to \dR_+ \cup \{\infty\}$ as follows.
For any profile $b_{-i} \in V^{n-1}$, define
\[ T_i(b_{-i}) = \inf \{b_i:A(b_i,b_{-i}) = i\} \]
to be the minimum bid of agent~$i$ that makes her win against $b_{-i}$;
If $A(b_i,b_{-i}) \neq i$ for all $b_i \in V_i$, then $T_i(b_{-i}) = \infty$.
By Lemma \ref{lemma:simple properties}(\ref{lem:MonoForSingleItem}), 
$A(b_i,b_{-i}) = i$ for every agent $i$, every bid profile $b_{-i}$ and any bid $b_i > T_i(b_{-i})$. 
The proof of Theorem~\ref{thm:ESIC-Char} is divided into eight steps:
\begin{itemize}
\item In Step~1 we prove that the winner at bid profile $b$ pays $T_i(b_{-i})$.
\item In Step~2 we prove that the amount the winner pays 
is some constant $t_i$, which depends only on her identity and not on the bid profile.
\item In Step~3 we prove that if at bid profile $b_{-i}$ agent $i$ has a bid that makes her win,
then for any bid profile $b'_{-i}$ that is dominated coordinatewise by $b_{-i}$, agent $i$ has a bid that makes her win.
\item In Step~4 we prove that if agent~$i$'s utility at bid profile $b$ is 0,
then agent~$i$'s utility is still 0 if agent $j \neq i$ lowers her bid.
\item In Step~5 we prove that if no agent wins, then the bid of each agent $i$ is at most $t_i$.
\item In Step~6 we prove that every pair of bid profiles $b$ and $b'$ that have the same set of agents who bid at least $t_i$,
such that at $b'$ at least one agent $j$ bids strictly more than $t_j$,
yield the same winner.
\item In Step~7 we define a priority ranking $R$ among the agents.
\item In Step~8 we finally prove that $M$ is a threshold mechanism with priority ranking $R$ and thresholds $(t_i)_{i \in I}$.
\end{itemize}

\bigskip
\noindent \textbf{Step 1:} 
For every agent $i$ and every bid profile $b$, if $A(b) = i$, then $P_i(b) = T_i(b_{-i})$.

Suppose first that $A(b) = i$ and $P_i(b) > T_i(b_{-i})$ for some agent $i$ and bid profile $b$.
Consider a bid $b'_i$ that satisfies $P_i(b) > b'_i > T_i(b_{-i})$.
As mentioned above, $A(b'_i,b_{-i}) = i$, and by Lemma~\ref{lemma:simple properties}(\ref{lem:equalPayments}), $P_i(b'_i,b_{-i}) = P_i(b_i,b_{-i})$.
Thus, $\UM_i((b'_i;b_{-i});b'_i) = b'_i - P_i(b'_i,b_{-i}) < 0$, which contradicts the IR condition.

Suppose now that $A(b) = i$ and $P_i(b) < T_i(b_{-i})$ for some agent $i$ and bid profile $b$.
Consider a bid $b'_i$ that satisfies $P_i(b) < b'_i < T_i(b_{-i})$.
Since $b'_i < T_i(b_{-i})$, we have $A(b'_i,b_{-i}) \neq i$.
However, $\UM_i(b;b'_i) = b'_i - P_i(b) > 0 = \UM_i((b'_i;b_{-i});b'_i)$, which contradicts the IC condition.
This concludes the proof of Step 1.
\hfill $\diamond$

\bigskip
\noindent \textbf{Step 2:} 
For every agent $i$, there exists $t_i \in \dR_+ \cup \{\infty\}$ such that for every bid profile $b_{-i} \in V^{n-1}$, either $T_i(b_{-i}) = t_i$ or $T_i(b_{-i}) = \infty$.

Assume, to the contrary, that there exists an agent $i$ and bid profiles $b_{-i}$ and $b'_{-i}$, 
such that $T_i(b_{-i}) < T_i(b'_{-i}) < \infty$.
Lemma~\ref{lemma:simple properties}(\ref{lem:MonoForSingleItem}) implies that there is a bid $b_i > \max\{T_i(b_{-i}), T_i(b'_{-i}) \}$ such that $A(b_i,b_{-i}) = A(b_i,b'_{-i}) = i$.
By Step 1, $P_i(b_i,b_{-i}) = T_i(b_{-i}) < T_i(b'_{-i}) = P_i(b_i,b'_{-i})$,
contradicting Lemma \ref{lem:stairESIC}.
%However, Lemma \ref{lem:stairESIC} suggests that $P_i(b_i,b_{-i}) = P_i(b_i,b'_{-i})$.
%Hence the contradiction.
\hfill $\diamond$
\bigskip

From now on we let $(t_i)_{i \in I}$ be the constants given by Step~2.
If $T_i(b_{-i}) = \infty$ for each $b_{-i} \in V^{n-1}$, or equivalently, if $A(b) \neq i$ for each $b \in V^n$, then we set $t_i = \infty$.
In other words, for any agent $i \in I$, we have $t_i \in \dR_+$ if and only if there exists a bid profile $b \in V^n$ such that $A(b) = i$.

\bigskip
\noindent \textbf{Step 3:} 
Fix a bid profile $b_{-i}$ such that $T_i(b_{-i}) = t_i < \infty$, and consider a bid profile $b'_{-i}$ such that $b'_j \le b_j$ for every agent $j \neq i$. 
Then, $T_i(b'_{-i}) = t_i$.

Fix a bid $b_i > t_i$, so that $A(b_i,b_{-i}) = i$, $P_i(b_i,b_{-i}) = t_i$, and $\UM_i((b_i,b_{-i});b_i) = b_i - t_i > 0$.
We will show that $A(b_i,b'_{-i}) = i$.
This will imply that $T_i(b'_{-i}) \leq T_i(b_{-i}) < \infty$,
and therefore
by Step~2 we will obtain that $T_i(b'_{-i}) = t_i$.

Assume then to the contrary that $A(b_i,b'_{-i}) \neq i$, and consider a single perturbation path $(b^r)_{r=0,1,\ldots,n}$, where $b^0 = (b_i,b_{-i})$ and $b^n = (b_i,b'_{-i})$.
Let $r$ be the minimal index such that $A(b^r) = i$ and $A(b^{r+1}) \neq i$.
By assumption,
$b^{r}_{r+1} = b_{r+1} \geq b'_{r+1} = b^{r+1}_{r+1}$.
If $A(b^{r+1}) = A(b'_{r+1},b^r_{-(r+1)}) = r+1$, then by Lemma \ref{lemma:simple properties}(\ref{lem:MonoForSingleItem}), $A(b^{r}) = A(b_{r+1},b^r_{-(r+1)}) = r+1$, contradicting the condition $A(b^r) = i$.
If $A(b^{r+1}) = j \notin \{r+1,i\}$, then given the bid profile $b^r$, agent $r+1$ can misreport $b'_{r+1}$, instead of $b_r$, thereby decreasing the utility of agent $i$. 
This contradicts the ESIC condition.
Hence, $A(b_i,b'_{-i}) = i$, and by Step $2$, 
$T_i(b'_{-i}) = t_i$.
This concludes the proof of Step $3$.
\hfill $\diamond$

\bigskip
\noindent \textbf{Step 4:} 
Let $b$ be a bid profile, let $j$ be an agent,
and let $b_j' < b_j$.
If $u_i(b;b_i)=0$ for every agent $i$, 
then $u_i((b'_j,b_{-j});b_i) = 0$ for every agent $i$ (including $i=j$).

Let $b$, $j$, and $b'_j$ be as in the claim,
and assume by contradiction that $u_i((b'_j,b_{-j});b_i) \neq 0$ for some agent $i$.
Since the mechanism is IR, it follows that $u_i((b'_j,b_{-j});b_i) > 0$.
By Lemma~\ref{lemma:simple properties}(\ref{lem:ZeroPayment}), $A(b')=i$, 
and hence by Lemma~\ref{lemma:simple properties}(\ref{lem:MonoForSingleItem}), $i\neq j$.
But then at the bid profile $b'$, agent~$j$ can misreport $b_j$, thereby lowering agent~$i$'s payoff while not affecting her own payoff.
This contradicts ESIC, and conclude the proof of Step~4.

\bigskip
\noindent \textbf{Step 5:} 
Let $b$ be a bid profile such that $A(b)=0$.
Then $b_i \leq t_i$ for every agent $i$.

Assume to the contrary that $A(b) = 0$ but there exists an agent $i$ such that $b_i > t_i$.
By the definition of $t_i$,
there exists a bid profile $b'$ such that $b_i'=b_i > t_i$ and $A(b')=i$.
Define the bid profile $c$ as follows:
$c_i=  b_i > t_i$ and $c_j < \min\{b_j,b'_j, t_j\}$ for every $j\neq i$; in case $\min\{b_j,b'_j,t_j\}=0$, set $c_j=0$.
Since $A(b) = 0$,
the IR condition implies that $u_i(b;b_i) = 0$ for every agent $i$.
By Step $4$, applied recursively for all agents $k$ such that $b_k > c_k$, we have $u_i(c;b_i)=0$.
Consider a single perturbation path $(b^r)_{r=0,1,\ldots,n}$ where $b^0 = b'$ and $b^n = c$.
Let $r$ be the first stage such that $A(b^r) = i$ and $A(b^{r+1}) \neq i$.
By the construction of $c$ and Step~1, 
$u_i(b^r;b_i) = b_i - t_i >0$ and $b^{r+1}_{r+1} = c_{r+1} < b'_{r+1} = b^r_{r+1}$.
Lemma~\ref{lemma:simple properties}(\ref{lem:MonoForSingleItem}) implies that $A(b^{r+1}) \not\in\{ i, r+1\}$.
Therefore, at the bid profile $b^r$, agent $r+1$ can misreport $c_{r+1}$ instead of $b^r_{r+1}$, and strictly decrease the utility of agent $i$ without affecting her own utility.
This violates the ESIC property, and concludes the proof of Step $5$.

\bigskip

For every bid profile $b$ define 
\[ I_b= \{j: b_j \geq t_j\} \subseteq I. \] 
This is the set of agents $j$ who bid at least $t_j$.

\bigskip
\noindent \textbf{Step 6:} 
For every two bid profiles $b$ and $b'$ such that 
(i) $I_b = I_{b'}$, 
(ii) $A(b)\neq 0$, and 
(iii) there exists an agent $j$ such that $b'_j>t_j$, 
we have $A(b)=A(b')$.

Assume to the contrary that the two bid profiles $b$ and $b'$ satisfy (i), (ii), and (iii)
but $i = A(b) \neq A(b')$.
By the definition of $t_i$, we have $i\in I_b = I_{b'}$.
By Step~5, 
$A(b') \neq 0$.
Denote $j = A(b')$.
By increasing $b_i$ and $b'_j$ if necessary and by Lemma \ref{lemma:simple properties}(\ref{lem:MonoForSingleItem}), we can assume w.l.o.g.~that $u_j(b';b'_j)= b'_j - t_j >0$ and $u_i(b;b_i)= b_i - t_i >0$.
We can also assume that $b_i'=b_i$, hence $u_j(b';b_j')= b_j - t_j >0$; indeed, otherwise, at bid profile $b'$ agent $i$ can misreport $b_i$, thereby decreasing the utility of agent $j$ without affecting her own utility.

Consider a single perturbation path $(b^r)_{r=0,1,\ldots,n}$ such that $b^0 = b$ and $b^{n} = b'$.
We will show that for each $r=0,1,\ldots,n-1$, 
if $A(b^r) = i$ then $A(b^{r+1}) = i$, which contradicts 
the facts that $A(b^0) = A(b) = i$ and
$A(b^n) = A(b') \neq i$.
Assume that there exists $r<n$ such that $A(b^r) = i$ and $A(b^{r+1}) \neq i$.
Since $b_i = b_i'$, we have $r+1 \neq i$.
Since $I_b = I_{b'}$, whether $r+1 \notin I_b$ or $r+1 \in I_b$, 
at the bid profile $b^r$,
when agent $r+1$ misreports $b_{r+1}^{r+1}$,
her utility does not decrease,
while the utility of agent $i$ strictly decreases (because $u_i(b^r;b_i)> 0 = u_i(b^{r+1};b_i)$), contradicting the ESIC property.
Thus, $A(b^{r+1}) = i$ as claimed.
This concludes the proof of Step $6$.
\hfill $\diamond$

\bigskip

We are now ready to define the priority ranking $R$ that will be used in the definition of the threshold mechanism.

\bigskip
\noindent \textbf{Step 7:} 
A definition of a priority ranking $R : I \to \{1,\ldots,n\}$.

Let $I_{\infty} = \{j \in I: t_j = \infty\}$ be the set of agents who never win.
To these agents, we arbitrarily assign ranking between $n-|I_{\infty}|+1,\ldots,n$.
That is, 
$R(I_\infty) = \{n-|I_{\infty}|+1,\ldots,n\}$.

If $I_\infty = I$, we are done with the definition of $R$.
Assume then that $I_{\infty} \neq I$,
and suppose by induction that we already defined the $k$ agents who have the highest priority, for some $k=0,1,\ldots,n-|I_{\infty}|-1$;
that is, we already defined $R^{-1}(1), R^{-1}(2),\ldots, R^{-1}(k)$.

Denote by $B_k$ the set of all bid profiles $b$ that satisfy the following properties:
\begin{itemize}
\item $b_i < t_i$ for each $i$ such that $R(i) \leq k$:
the bids of agents with high rank is low.
\item $b_i \geq t_i$ for each $i$ such that $i \not\in I_\infty \cup \{R^{-1}(1), R^{-1}(2),\ldots, R^{-1}(k)\}$,
with at least one strict inequality.
That is, agents who were not ranked yet have high bids.
\end{itemize}

By Step~6,
for every $b,b' \in B_k$ we have $A(b) = A(b')$.
By Step~5, it cannot be that $A(b) = 0$ for each $b \in B_k$.
Since $b_i < t_i$ for every $i$ such that $R(i) \leq k$,
we have $A(b) \not\in \{R^{-1}(1), R^{-1}(2),\ldots, R^{-1}(k)\}$ for each $b \in B_k$.
It follows that there is an agent $i_{k+1} \not\in I_\infty \cup \{R^{-1}(1), R^{-1}(2),\ldots, R^{-1}(k)\}$ such that $A(b) = i_{k+1}$ for every $b \in B_k$.
We let $i_{k+1}$ be the next agent according to $R$,
so that $R(i_{k+1}) = k+1$. 

\bigskip
\noindent\textbf{Step 8:} 
$M$ is a threshold mechanism with the priority ranking function $R$ defined in Step~7 and the thresholds $(t_i)_{i \in I}$ given by Step~2.

We will show that there is a threshold mechanism $M^* = (A^*,P^*)$ with the priority ranking function $R$ and the thresholds $(t_i)_{i \in I}$
such that for every bid profile $b \in V^n$,
$A(b)$ and $P(b)$ coincide with the agent and payments indicated by $M^*$.

By Steps~1 and~2,
if $A(b) = i$ then $P^*(b) = t_i = P(b)$,
while by Lemma~\ref{lemma:simple properties}(\ref{lem:ZeroPayment}), if $A(b) \neq i$, then $P^*(b) = 0 = P(b)$.
Hence the payment function in $M$ coincides
with the payment function of any threshold mechanism with the priority ranking function $R$ and the thresholds $(t_i)_{i \in I}$.
We turn to handle the allocation function.

If $I_b = \emptyset$, 
that is, every agent $i$ bids below $t_i$,
then the definition of $(t_i)_{i \in I}$ implies
that $A(b) = 0$,
Moreover, any threshold mechanism $M^* = (A^*,P^*)$ with the thresholds $(t_i)_{i \in I}$ will satisfy in this case $A^*(b) = 0$.

Assume from now on that $I_b \neq \emptyset$,
and denote by $i^* \in I_b$ the agent with minimal ranking in $I_b$ according to $R$:
\[ R(i^*) < R(i), \ \ \ \forall i \in I_b. \]
Then, for any threshold mechanism with the priority ranking function $R$ and the thresholds $(t_i)_{i \in I}$,
\begin{itemize}
\item If there is $j \in I_b$ such that $b_j > t_j$, 
then $A^*(b) = i^*$.
\item Otherwise, $A^*(b) \in I_b \cup \{0\}$.
\end{itemize}
We will show that in this case $A(b)$ coincides with the above specifications.
By the definition of $I_\infty$ we have $A(b) \not\in I_\infty$,
and since $M$ is IR we have $A(b) \in I_b \cup \{0\}$.
Hence if $b_j = t_j$ for every $j \in I_b$, we have $A(b)  \in I_b \cup \{0\}$.
It therefore remains to handle the case that $b_j > t_j$ for at least one agent $j \in I_b$.

Assume then that this is the case.
By Step $5$, $A(b) \neq 0$.
Denote $i = A(b)$.
We need to show that $i = i^*$.
By Step~1 and since $M$ is IR, we have $b_i \ge t_i$, so that $i \in I_b$.
Therefore, if $i \neq i^*$, then $R(i^*) < R(i)$.

Because $R$ is a well-defined priority ranking, there is
a bid profile $b^*$ such that 
(i) $A(b^*)=i^*$, and 
(ii) $b^*_j\geq t_j$ if and only if $R(j)\geq R(i^*)$ and $j \not\in I_\infty$.

If $I_b = I_{b^*}$, then by Step $6$ we have $i=A(b)=A(b^*)=i^*$, as needed.
Otherwise, $I_b \subsetneq I_{b^*}$. 
Consider the bid profile $c$ where $c_j = b_j$ for every $j\in I_{b^*} \setminus I_b$, and $c_j = b^*_j$ otherwise.
The bid profile $c$ is generated from $b^*$ by reducing the bid of every agent $j\in I_{b^*} \setminus I_b $ from $b^*_j$ to $b_j$.
By Step $3$, $A(c)=A(b^*)=i^*$, which implies that $I_c=I_b$ by construction.
Thus, by Step $6$ we have $i^* = A(c)= A(b)= i$, as needed.
\hfill
\end{proof}

\subsection{A supporting result for at least 3 agents}

The next result implies that for IR and SIC mechanisms, when there are at least three agents $i$, $j$, and $k$, and agent $i$ increases her bid while all other agents keep their bids fixed, it is impossible that the winner changes from $j$ to $k$.

\begin{proposition}\label{prop:OnlyOne}
Suppose the mechanism $M$ is \emph{IR} and \emph{SIC} with $n \ge 3$.
Then, there are no three distinct agents $i,j,k$ and bid profiles $b$, $b' = (b'_i,b_{-i})$, and $b'' = (b''_i,b_{-i})$, such that $b''_i > b'_i > b_i$, 
$u_j(b;b_j) > 0$, $u_k(b';b_k) > 0$, and $u_k(b'';b_k) > 0$.
\end{proposition}

\begin{proof}

\noindent\textbf{Step 0:} Preparations.

Assume to the contrary that there are three distinct agents $i,j,k$ and bid profiles $b$, $b' = (b'_i,b_{-i})$, and $b'' = (b''_i,b_{-i})$, such that $b''_i > b'_i > b_i$,
$u_j(b;b_j) > 0$, $u_k(b';b_k) > 0$, and $u_k(b'';b_k) > 0$.

Since $u_j(b;b_j) > 0$, we have $A(b) = j$, and since $u_k(b';b_k) > 0$, we have $A(b') = k$. 
Let $t_k = T_k(b'_{-k})$, and consider the profiles $c = (t_k,b_{-k})$ and $c' = (t_k,b'_{-k})=(t_k,b'_i,b_{-\{i,k\}})$.
See Figure \ref{fig:transition between bid profile b b' c c'} for the relations between these profiles.

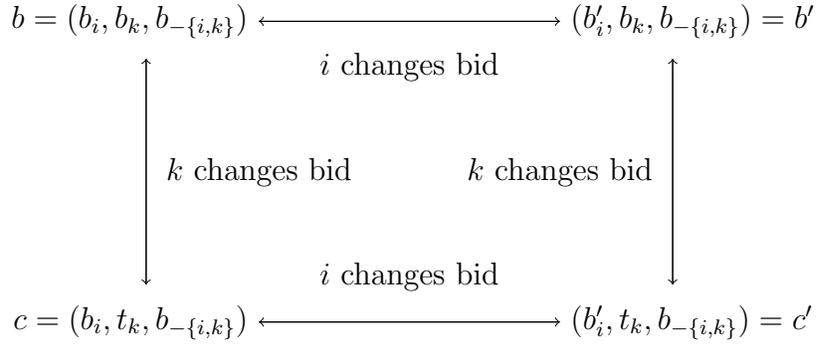
\begin{figure}
    \centering
\begin{tikzpicture}
[cross/.style={path picture={
\draw[black]
(path picture bounding box.south east) -- (path picture bounding box.north west) (path picture bounding box.south west) -- (path picture bounding box.north east);
}}]
\tikzstyle{dot}=[circle,draw,inner sep=1.5,fill=black]
\draw[->] (0,0) -- (4, 0) node[right] {$(b'_i,t_k,b_{-\{i,k\}}) = c'$};
\draw[->] (4,0) -- (0, 0) node[left] {$c = (b_i,t_k,b_{-\{i,k\}})$};
\draw[->] (0,4) -- (4, 4) node[right] {$(b'_i,b_k,b_{-\{i,k\}}) = b'$};
\draw[->] (4,4) -- (0, 4) node[left] {$b = (b_i,b_k,b_{-\{i,k\}})$};

\draw[->] (-1.5,0.5) -- (-1.5, 3.5);
\draw[->] (-1.5,3.5) -- (-1.5, 0.5);
\draw[->] (5.5,0.5) -- (5.5, 3.5);
\draw[->] (5.5,3.5) -- (5.5, 0.5);

\draw (2,0.6) node{$i$ changes bid};
\draw (2,3.4) node{$i$ changes bid};
\draw (0,2) node{$k$ changes bid};
\draw (4,2) node{$k$ changes bid};

\end{tikzpicture}

\caption{The relations between bid profiles $b$, $b'$, $c$, and $c'$.}
    \label{fig:transition between bid profile b b' c c'}
\end{figure}

Denote $t'_i = T_i(c'_{-i})$. 

\bigskip\noindent\textbf{Step 1:}  $u_m(c';c_m') = 0$ for all agents $m$.

\begin{itemize}
    \item If $A(c') = 0$, then by IR the utilities of all agents is 0: $u_m(c') = 0$ for all agents $m$. 
    \item If $A(c') = k$, then Lemma~\ref{lemma:simple properties}(\ref{lem:equalPayments}) implies that the payment of agent $k$ is $T_k(b'_i, b_{-\{i,k\}}) = T_k(b'_{-k}) = t_k$, and hence $u_k(c';c'_k) = 0$.
    \item If $A(c') = l \neq k$ and $u_l(c';c'_l) > 0$, then at $c'$ agent $k$ can misreport $b_k$ (instead of $t_k$), and so the bid vector after the deviation is $b'$. Since $u_m(b';b'_m) = 0 = u_m(c';c'_m)$ for every $m \neq l$, and $u_l(c';c'_l) > 0 = u_l(b';b'_l)$, it follows that the mechanism is not SIC at $b'$, a contradiction. Hence, at the profile $c'$, either the item is unassigned, or the winner of the item gets 0 utility. 
\end{itemize}
In all cases above, $u_m(c';c_m') = 0$ for all agents $m$.

\bigskip\noindent\textbf{Step 2:} $t'_i = b'_i$.

Assume to the contrary that $t'_i \neq b'_i$.
If $b'_i > t'_i$, then Lemma~\ref{lemma:simple properties}(\ref{lem:MonoForSingleItem}) implies that $A(c') = i$, and hence $u_i(c';b_i') = b'_i - t'_i > 0$, which contradicts Step 1.
Hence $b'_i < t'_i$, and therefore by the definition of $t'_i$ we have $A(c') \neq i$.
Since $b_i < b'_i$, by Lemma \ref{lemma:simple properties}(\ref{lem:MonoForSingleItem}), we have $A(c) \neq i$.

We next argue that either $A(c) = 0$, or the winner $A(c)$ has utility 0 at $c$.
Indeed, suppose that $l = A(c) \neq 0$ and $u_l(c;c_l) > 0$.
We will derive a contradiction to the assumption that $M$ is SIC.
Note that by the previous paragraph, $l \neq i$, and hence $c_l = c'_l$.
At bid profile $c$, agent $i$ can misreport from $b_i$ to $b'_i$, changing the bid profile to $c'$. 
Since $u_{l}(c;c_l) > 0 = u_{l}(c';c'_l)$ and $u_m(c';c_l) = 0 = u_m(c';c_l)$ for each $m \neq A(c)$,
we reach a contradiction to the SIC assumption.

We are now ready to derive the contradiction to the conclusion that $b'_i < t'_i$.
Suppose that at the bid profile $b$, agent $k$ misreports from $b_k$ to $t_k$, resulting in the profile $c$. 
The utility of agent $j$ strictly decreases, while the utility of every other agent remains the same, contradicting the assumption that $M$ is SIC.

\bigskip\noindent\textbf{Step 3:} The end of the proof.

The argument in Step~2 does not use the bid vector $b''$.
Hence, an analogous argument to that provided in Step~2 shows that $b''_i = t'_i$. 
But then $b'_i = t'_i = b''_i$, which contradicts the assumption that $b''_i > b'_i$. 
\end{proof}

\subsection{Every IR and SIC mechanism is ESIC}

In this section we prove the last part of the proof of Theorem~\ref{thm:charecterization}.

\begin{theorem}
\label{thm:SIC}
Every \emph{IR} and \emph{SIC} mechanism is \emph{ESIC}.
\end{theorem}

\begin{proof}
If $n=2$, then SIC and ESIC mechanisms coincide by definition.
Assume then that $n\geq 3$, and suppose by contradiction that there exists an IR and SIC mechanism $M$ which is not a threshold mechanism.
By Theorem \ref{thm:ESIC-Char}, $M$ is not ESIC.
Hence, there exists a bid profile $b = (b_1,\ldots,b_n)$, an agent $i$, 
and a bid $b'_i \neq b_i$ such that either 
(i) $u_i((b'_i,b_{-i});b_i) > u_i(b;b_i)$ or 
(ii) $u_i((b'_i,b_{-i});b_i) = u_i(b;b_i)$ and there exists an agent $j$ such that $u_j((b'_i,b_{-i});b_j) < u_j(b;b_j)$.
Since $M$ is SIC, (i) cannot hold,
and hence (ii) holds.
Moreover, since $M$ is SIC, there exists an agent $k \neq i,j$ such that $u_k((b'_i,b_{-i});b_k) > u_k(b;b_k)$.
Without loss of generality, assume that $b_i < b'_i$.

Let $b' = (b'_i,b_{-i})$.
Since $M$ is IR, we have $u_j(b;b_j) > u_j((b'_i,b_{-i});b_j) \ge 0$. 
This implies that $A(b) = j$.
Similarly, $u_k((b'_i,b_{-i});b_k) > u_k(b;b_k) \ge 0$, which implies that $A(b'_i,b_i) = k$.

If there exists a bid $b_i''$ such that $u_l((b''_i,b_{-i}); b_l) = 0$ for all agents $l$, 
then $(b''_i,b_{-i})$ is a deviation from $(b'_i,b_{-i})$ that contradicts SIC.
%this would contradict the SIC assumption.
Thus, for every $b_i'' \notin \{ b_i,b_i'\}$, either $A(b''_i,b_{-i}) = i$ or the winner $A(b''_i,b_{-i})$ has a strictly positive utility.

Let $t_i = T_i(b_{-i})$. 
By definition and Lemma~\ref{lemma:simple properties}(\ref{lem:MonoForSingleItem}), for any $b''_i > t_i$ we have $A(b''_i,b_{-i}) = i$, and for any $b''_i < t_i$ we have $A(b''_i,b_{-i}) \neq i$.
Hence, $b_i < b_i' \leq t_i$.
If $b_i' < t_i$, then $(b'_i,t_i)$ is an open (non-empty) interval, and there exist two distinct bids  $b''_i, b'''_i \in (b'_i,t_i)$ that yield the same winner: $A(b''_i,b_{-i}) = A(b'''_i,b_{-i})$. 
Because this winner has a strictly positive utility, we get a contradiction to Proposition \ref{prop:OnlyOne}, so we conclude that $b_i < b'_i = t_i$.

Since $t_i = b'_i$ and $u_k((b'_i,b_{-i});b_k) > 0$,  
at the profile $(b'_i,b_{-i})$ agent $i$ can misreport a bid higher than $b'_i$, so that in the new profile agent $i$ wins the item and pays $t_i$ (as in Step $1$ of Theorem \ref{thm:ESIC-Char}), so her utility remain zero. 
In this case, the utility of agent $k$ strictly decreases, and the utility of every other agent remains $0$, contradicting the SNE assumption.
\hfill 
\end{proof}

\section{The optimal spite free mechanism in Section~\ref{Section: optimal spite free mechanism - example}} \label{Appendix: optimal spite free mechanism - example}

We here calculate the optimal threshold mechanism for the case where the agents are symmetric, 
and their private values are i.i.d.~and uniformly distributed on $[0,1]$.

The expected revenue is given by
\[ \gamma(t_1,\dots,t_n) = 
(1-t_n)t_n + t_n\Bigl( (1-t_{n-1})t_{n-1} + t_{n-1}\bigl( (1-t_{n-2})t_{n-2} + t_{n-2}(\dots (1-t_1)t_1\bigr)\Bigr).\]
Indeed, with probability $1-t_n$ the private value of the highest-ranked agent, agent $n$,
exceeds her threshold $t_n$, in which case the seller's revenue is $t_n$;
with probability $t_n$ the private value of agent $n$ is below $t_n$, and then by induction the expected revenue is given by the term that multiplies $t_n$ in the second summand.

The term $t_1$ appears only in the last term, as $(1-t_1)t_1$,
and hence the optimal threshold for agent~1 is
$t_1^* = \frac{1}{2}$.
The term $t_2$ appears only in the event that agents $n,n-1,\dots,3$ did not win the item. 
In that case, it appears as
\[ (1-t_2)t_2 + t_2(1-t^*_1)t^*_1  = t_2\bigl(1+(t^*_1)^2 - t_2\bigr),\]
where the equality holds because $t^*_1 = 1-t^*_1 = \frac{1}{2}$.
The roots of this function are $t_2 = 0$ and $t_2 = 1+(t^*_2)^2$,
hence
\begin{equation}
\label{equ:41}
t_2^* = \frac{1+(1-t^*_1)t^*_1}{2} = \frac{1+(t^*_1)^2}{2}.
\end{equation} 
The term $t_3$ appears only in the event that agents $n,n-1,\dots,4$ did not win the item. 
In that case, it appears as
\begin{align*}
(1-t_3)t_3 + t_3 \bigl( (1-t^*_2)t^*_2 + t^*_2(1-t^*_1)t^*_1 \bigr)
&= t_3 \bigl( 1 - t_3 + (1-t^*_2)t^*_2 + t^*_2(1-t^*_1)t^*_1 \bigr).
\end{align*} 
This function is quadratic in $t_3$,
one of its root is $0$,
and hence its maximum is attained at half the second root, which, by Eq.~\eqref{equ:41}, is
\begin{align*}
t^*_3 = \frac{1 + (1-t^*_2)t^*_2 + t^*_2(1-t^*_1)t^*_1}{2}= \frac{1-(t^*_2)^2}{2} + (t^*_2)^2
= \frac{1+(t^*_2)^2}{2}.
\end{align*} 
Continuing recursively in an analogous manner, we obtain the recursive relation 
\[ t^*_{i+1} = \frac{1+(t^*_{i})^2}{2}, \ \ \ \forall i = \{1,2,\dots,n-1\}. \]

\section{Proof of Theorem \ref{Theorem - sequential mechanism}} \label{Appendix - proof of Theorem for sequential}

Let $M$ be a sequential mechanism with thresholds $(t_i)_{i\in I}$.
Let $v=(v_i)_{i\in I}$ be the true valuations.
The marginal contribution of the $q$'th item to agent~$i$'s valuation is
$\Delta_i(q)=v_i(q)-v_i(q-1)$.
By convention, $\Delta_i(K+1) = -1$.
Hence,
$v_i(q)=\sum_{\ell=1}^q\Delta_i(\ell)$,
and since agent~$i$'s valuation is submodular,
$\Delta_i(1)\ge \Delta_i(2)\ge\cdots\ge \Delta_i(K)$.

Denote by $q_i$ the quantity that maximizes agent~$i$'s utility, so that $\Delta_i(q_i) \geq t_i > \Delta_i(q_i+1)$.
Assume w.l.o.g.~that the priority ranking is the identity function, so that agent~$1$ is the highest ranked.
The mechanism $M$ allocates $Q_1 = \min\{ K, q_1\}$ to agent~$1$, $Q_2 = \min\{K-Q_1,q_2\}$ to agent~$2$, and, more generally, $Q_i = \min\{K-\sum_{j=1}^{i-1}Q_j,q_i\}$ to each agent~$i$.

Since the number of items allocated to each agent~$i$ is at most $q_i$,
and since agent~$i$'s valuation is submodular,
$M$ is IR.
Since in addition agent~$i$'s report does not affect the allocation of agents $1,2,\dots,i-1$, $M$ is IC.

To see that $M$ is SIC, note that for each $i$, the utility of each agent $j \geq i$ is non-decreasing in the amount of items not allocated for the first $i$ agents, $\widehat K_i = K-\sum_{l=1}^{i-1}Q_l$.
Indeed, if $\widehat K_i = 0$, then each agent $j \geq i$ is not allocated any item, and hence, since $M$ is IR, increasing the amount of items allocated to these agents cannot harm them.
Suppose then that $\widehat K_i > 0$.
The number of items $M$ allocated to agent~$i$ is  $Q_i = \min\{\widehat K_i,q_i\}$.
\begin{itemize}
\item
If $Q_i = q_i$,
agent~$i$ is allocated $q_i$ items also after increasing the supply.
Therefore, the excess supply is transferred to agents $i+1,i+2,\dots,n$. 
By induction, the utilities of these agents are non-decreasing in the supply.
\item 
If $Q_i = \widehat K_i$,
then no items are left for agents $i+1,i+2,\dots,n$. 
Therefore, when the supply increases agent~$i$ may obtain more items,
in which case her utility does not decrease.
And the following agents cannot lose, since $M$ is IR.
\end{itemize}

We can now conclude that $M$ is SIC.
Indeed, when agent~$i$ misreports her valuation,
the allocation of agents $1,2,\dots,i-1$ is not affected.
Recall that $Q_i = \min\{K-\sum_{j=1}^{i-1}Q_j,q_i\}$.
\begin{itemize}
\item
If $Q_i = q_i$,
then,
since agent~$i$'s valuation is submodular,
any misreport that increases the amount of items allocated to agent~$i$ lowers her utility.
And a misreport that lowers the amount of items allocated to agent~$i$ increases the amount of items left to agents $i+1,\dots,n$,
and hence, by the above discussion,
cannot lower their utilities.
\item
If $Q_i = K-\sum_{j=1}^{i-1}Q_j$,
then agent~$i$ is allocated all items that are not allocated to agents $1,\dots,i-1$.
In particular, all subsequent agents are allocated no item.
Since $M$ is IR,
any misreport cannot lower the utility of subsequent agents.
\end{itemize}

\section{Proof of Theorem \ref{thm:MultipleChar}} \label{Appendix - Proof of Theorem 2}

\noindent\textbf{Proof of Part 1.}
Assume by contradiction that $A_i(b)\neq T_i$ while
\[
b_i(T_i)-P_i^{T_i}(b_{-i}) > b_i(S_i)-P_i^{S_i}(b_{-i}),\quad\forall\,S_i\neq T_i.
\]
By definition of $W_{-i}^{T_i}$, 
there exists $b'_i$ with $A_i(b'_i,b_{-i})=T_i$.
By Lemma~\ref{lem:payment}, $P_i(b'_i,b_{-i})=P_i^{T_i}(b_{-i})$.
Hence,
\[
u_i\big((b'_i,b_{-i});b_i\big)=b_i(T_i)-P_i^{T_i}(b_{-i}) > b_i(S_i)-P_i^{S_i}(b_{-i}) \ge u_i\big((b_i,b_{-i});b_i\big),
\]
and therefore $b'_i$ is a profitable deviation at $b$, contradicting SIC. Thus $A_i(b_i,b_{-i})=T_i$.

\bigskip
\noindent\textbf{Proof of Part 2.}
Let $A_i(b)=T_i$ and assume, by contradiction, that there exists $S_i\neq T_i$ with
$b_i(S_i)-P_i^{S_i}(b_{-i}) > b_i(T_i)-P_i^{T_i}(b_{-i})$.
This inequality implies that $b_{-i} \in W_{-i}^{S_i}$, 
and hence there exists $b''_i$ with $A_i(b''_i,b_{-i})=S_i$.
By Lemma~\ref{lem:payment}, $P_i(b''_i,b_{-i})=P_i^{S_i}(b_{-i})$, hence
\[
u_i\big((b''_i,b_{-i});b_i\big)=b_i(S_i)-P_i^{S_i}(b_{-i}) > b_i(T_i)-P_i^{T_i}(b_{-i}) = u_i\big((b_i,b_{-i});b_i\big).
\]
Therefore, $b''_i$ is a profitable deviation at $b$, contradicting SIC.

\bigskip
\noindent\textbf{Proof of Part 3.}
Assume $I=\{i,j\}$,
so that $b_{-i}$ translates to $b_j$. 
Fix a bundle $T_i \subseteq S$.
We argue that if $b_j \neq b'_j$ and $A(b_i,b_j) = A(b_i,b'_j)$,
then $P(b_i,b_j) = P(b_i,b'_j)$.
Indeed, $P_j(b_i,b_j) = P_j(b_i,b'_j)$ by Lemma~\ref{lem:payment},
and $P_i(b_i,b_j) = P_i(b_i,b'_j)$ because otherwise agent~$j$ would have a deviation that does not change her payoff but harms agent~$i$, contradicting SIC: 
a deviation to $b'_j$ at $(b_i,b_j)$
(if $P_i(b_i,b_j) < P_i(b_i,b'_j)$)
or to $b_j$ at $(b_i,b'_j)$
(if $P_i(b_i,b_j) > P_i(b_i,b'_j)$).

Hence,
the number of distinct payments 
when agent~$j$ obtains $T_i$ is at most $2^{K-|T_i|}$,
the number of subsets of the complement of $T_i$. 

\end{appendices}

\end{document}